\pdfoutput=1
\RequirePackage[english]{babel}
\documentclass[a4paper,english]{amsart}

\usepackage[sc,osf]{mathpazo}
\DeclareTextFontCommand{\textsl}{\fontfamily{ppl}\fontshape{sl}\selectfont}
\usepackage[euler-digits]{eulervm}
\usepackage{microtype}

\usepackage{amssymb,amscd,amsthm}
\usepackage[foot]{amsaddr}
\usepackage{mathtools}
\usepackage[all]{xy}
\usepackage[
pdftitle={On the computational complexity of a game of cops and robbers},
pdfauthor={Marcello Mamino},
pdfkeywords={computational complexity, cops and robbers, pursuit games},
colorlinks=true,
citecolor=blue,
linkcolor=blue,
urlcolor=red]{hyperref}

\title[Complexity of cops and robbers]{On the computational complexity of\\a game of cops and robbers}
\author[M.~Mamino]{\lsstyle Marcello Mamino}
\address{\textsc{cmaf} Universidade de Lisboa\\
Av. Prof. Gama Pinto 2\\
1649-003 Lisboa, Portugal}
\email{\mailto{mamino@ptmat.fc.ul.pt}}
\thanks{Partially supported by Funda\c{c}\~ao para a Ci\^encia e a Tecnologia grant SFRH/BPD/73859/2010}

\date{\mydate{27}{ii}{2012}, \textit{revised}: \mydate{24}{xi}{2012}}
\subjclass[2000]{\textsc{91a24, 91a43, 68q17}}

\makeatletter
\def\Hy@Warning#1{}
\makeatother

\makeatletter
\def\@setemails{% 
\ifnum\theg@author > 1 
\mbox{{\itshape E-mail addresses}:\space}{\ttfamily\emails}. 
\else 
\mbox{{\itshape E-mail address}:\space}{\ttfamily\emails}. 
\fi% 
}
\makeatother

\makeatletter
\def\ps@plain{\ps@empty
  \def\@oddfoot{\normalfont\normalsize \hfil\thepage\hfil}%
  \let\@evenfoot\@oddfoot}
\def\ps@firstpage{\ps@plain
  \def\@oddfoot{\normalfont\normalsize \hfil\thepage\hfil
     \global\topskip\normaltopskip}%
  \let\@evenfoot\@oddfoot
  \def\@oddhead{\@serieslogo\hss}%
  \let\@evenhead\@oddhead % in case an article starts on a left-hand page
}
\def\ps@headings{\ps@empty
  \def\@evenhead{%
    \setTrue{runhead}%
    \normalfont\normalsize
    \rlap{\thepage}\hfil
    \textsc{\lsstyle\MakeLowercase{\shortauthors}\hfil}}%
  \def\@oddhead{%
    \setTrue{runhead}%
    \normalfont\normalsize \hfil
    \textsc{\lsstyle\MakeLowercase{\rightmark{}{}}}\hfil\llap{\thepage}}%
  \let\@mkboth\markboth
}
\makeatother
\def\mathshift{$}
\catcode`$=13
\def${\ifinner\expandafter\mathshift\else\expandafter\myshift\fi}
\def\myshift#1${\raisebox{0ex}[0ex][0ex]{\mathshift#1\mathshift}}
\AtBeginDocument{\catcode`$=13}

\makeatletter
\let\oldsection\section
\def\newsection#1{\oldsection{\lsstyle #1}}
\def\newsectionr#1{\oldsection*{\lsstyle #1}}
\def\section{\@ifstar\newsectionr\newsection}
\makeatother

\makeatletter
\newtheorem{ghost@theorem}{}[section]
\def\@maketheorem#1=#2;{
	\newtheorem{#1}[ghost@theorem]{#2}}
\def\maketheorem#1{
	\@for\@x:=#1\do{
		\expandafter\@maketheorem\@x;}}
\makeatother

\let\oldexists\exists\def\exists{\oldexists\mkern 1mu}

\let\n\oldstylenums
\def\mydate#1#2#3{\hbox{\n{#1}$\cdot${\sc#2}$\cdot$\n{#3}}}

\def\-{\nobreakdash-\hspace{0pt}}
\def\U-{\raise0.2ex\hbox{-}}
\def\url#1{\href{#1}{url\nobreakdash---\texttt{#1}}}

\def\mailto#1{\href{mailto:#1}{\texttt{#1}}}
\def\margin#1{\hbox to 0pt{\hss#1\ }}
\def\eatspace#1{#1}
\makeatletter
\def\myitem#1{%
\hfil\break\margin{\textsc{\hbox to 1em{\hss #1\hss}-}}%
\def\@currentlabel{\textsc{#1}}\eatspace}
\makeatother

\theoremstyle{plain}
\maketheorem{%
	theorem=Theorem,%
	proposition=Proposition,%
	lemma=Lemma,%
	corollary=Corollary,%
	fact=Fact%
}
\theoremstyle{definition}
\maketheorem{%
	definition=Definition,%
	notation=Notation,%
	example=Example%
}
\theoremstyle{remark}
\maketheorem{%
	observation=Observation,%
	remark=Remark%
}

\parskip=\baselineskip

\makeatletter
\@for\theoremstyle:=definition,remark,plain\do{%
\expandafter\g@addto@macro\csname th@\theoremstyle\endcsname{%
\setlength\thm@preskip\parskip%
\setlength\thm@postskip{0pt}%
}}

\renewenvironment{proof}[1][\proofname]{\par
  \pushQED{\qed}%
  \normalfont \topsep=0pt
  \trivlist
  \item[\hskip\labelsep
        \itshape
    #1\@addpunct{.}]\ignorespaces
}{%
  \popQED\endtrivlist\@endpefalse
}

\def\section{\@startsection{section}{1}%
  \z@\z@{1sp}%
  {\normalfont\scshape\centering}}
\makeatother

\makeatletter
\edef\orig@output{\the\output}
\output{\setbox\@cclv\vbox{\unvbox\@cclv\vspace{0pt plus 20pt minus 20pt}}\orig@output}
\makeatother

\begin{document}

\begin{abstract}
We study the computational complexity of a perfect-information two-player
game proposed by Aigner and Fromme in~\cite{AiFro84}. The game takes place
on an undirected graph where $n$ simultaneously moving cops attempt to
capture a single robber, all moving at the same speed. The players are
allowed to pick their starting positions at the first move.  The question
of the computational complexity of deciding this game was raised in
the~'90s by~Goldstein and Reingold~\cite{GoRe95}. We prove that the game
is hard for~{\rm PSPACE}.
\end{abstract}

\maketitle

\section{Introduction}

\def\CR{\textsc{C\&R}}
We consider a two-player perfect-information game, \textsc{Cops and
Robber} (\CR), in which a given number~$n$ of cops attempts to capture a
single robber by moving over the edges of an undirected graph. At the
first move, the player controlling the cops chooses starting vertices for
the $n$ cops, then the robber's player chooses his vertex. After that, the
cops and the robber are moved in alternating turns from one vertex to an
adjacent vertex, all the cops move simultaneously. The game ends when one
of the cops captures the robber. We prove in Theorem~\ref{th-main} that
the problem of determining which player has a winning strategy in~\CR\ is
hard for~{\rm PSPACE}.

\CR\ has been first considered by~Aigner and~Fromme in~\cite{AiFro84}
generalizing a game that was proposed by~Nowakowski and~Winkler
in~\cite{NoWi83} and independently by~Quilliot in~\cite{Quil83}. Since
then, this game has been the object of intense study from the
combinatorial point of view: for a survey see~\cite{FoThi08}, an
up-to-date account of the state of the art can be found in the recent book
of Bonato and Nowakowski~\cite{Bonato-Nowakowski}. From the
point of view of computational complexity, the first study appears
in~\cite{GoRe95}, where~Goldstein and Reingold prove that versions of the
game played on directed graphs or starting from a fixed initial position
are complete for~{\rm EXPTIME}. In the same paper, the unrestricted game
is conjectured to be complete for~{\rm EXPTIME}.

Although the complexity of many similar games has been determined (the
reader is referred again to~\cite{FoThi08} for bibliography), the only
lower bound for the game with elective initial positions on an undirected
graph appeared only recently (for some positive algorithmic results
see~\cite{BeInt93,GoRe95,HaMa06}).  Namely, in~\cite{Foetal10} it is shown that
the problem of determining whether the cops have a winning strategy in a
given instance of~\CR\ is hard for~{\rm NP}.
The difficulty in proving any complexity result for~\CR\ lies in the
extremely dynamic nature of the game.
For example, the reader may observe that the
cops can not really make any mistake (on a connected graph), since from whatever
position they reach, they can go back to the initial position and restart the game.
Arguably, this makes it very hard to force the cops to make any given move.
As expected, the complexity of the game
is much easier to assess when one adds further constraints on the set of
possible moves, as in~\cite{GoRe95} or~\cite{FoGoLo12}. In fact, we do
precisely the same, reducing to~\CR\ a new and more flexible pursuit
game, which nevertheless is purely an extension of~\CR---\textit{i.e.} all
instances of~\CR\ are also instances of our game, see
Section~\ref{sec-CRp}. Our game greatly simplifies the proof of existing
results -- see Corollary~\ref{th-NP} -- and also allows us to almost
simulate~\CR\ played on a directed graph, a game which is known to be {\rm
EXPTIME-}complete.

{\rm NP}-hardness of~\CR\ is proven in~\cite{Foetal10},
and also in our Corollary~\ref{th-NP}, both proofs through DOMINATING-SET. Although
the arguments are based on a slightly different concept, both produce a graph in which
no real action takes place: if the robber is captured, capture occurs at
the first move. In this case, the only difficulty for the cops lies in the
choice of their initial position, and the game is indeed trivial for the
robber. So, apparently, to the researcher trying to prove {\rm NP}-hardness, the
elective initial positions are more of an asset than a hindrance.
However, as soon as one tries to attain hardness for higher
complexity classes, the game can not be made trivial for the robber any
more, and the initial choice of the cops is not hard enough to be
exploitable. On the contrary, in this situation, the unpredictability of the
initial positions becomes the main obstacle, and a different
technique must be used.

In~\cite{GoRe95}, \cite{FoGoLo12}, and~\cite{Foetmanyal12}, {\rm PSPACE-}~and {\rm EXPTIME-}hardness
results for variants of~\CR\ are proven by reduction of games played on
boolean formul\ae. The typical method is to simulate operations on boolean
variables by the action of several gadgets, which the robber is forced to
traverse in a well determined order during his flight. At the end of the
simulated boolean game, the robber has a chance to reach perpetual safety.
The most immediate type of perpetual safety comes in the form of a safe subgraph
(a subgraph in which the cops can never capture the robber, whatever the
initial position), however the presence of safe subgraphs is incompatible
with the players choosing their own initial positions, since the robber could
just pick his vertex in it.  In~\cite{FoGoLo12} is discussed a variant
of~\CR~\textit{without recharging}, \textit{i.e.} imposing a constraint on
the maximal number of moves that each cop can make, and, in the same
paper, {\rm PSPACE-}completeness is proven for this variant. Here, the
uncertainty over the initial position is dealt with observing that no cop
can start too far away from any vertex of the graph, since at least one cop
must be able to reach any vertex in case the robber shows up there, hence
the positions of the cops can be forced by connecting long antenn\ae\ to
the vertices where we want a cop to be placed. A complicated collection of
gadgets can then be devised to constrain the whole initial position and
simulate~QBF (evaluation of Quantified Boolean Formul\ae). For directed graphs, the idea of~\cite{GoRe95} is to start
with a construction having fixed initial positions, in which the robber
reaches a safe subgraph if he wins the simulated boolean game. Then, the
constructed graph is modified by substituting the safe subgraph by a reset
mechanism, which is safe just as long as the cops are not precisely in
their initial positions. At this time, the robber is forced to exit the
reset mechanism and play the boolean game. The directness
of the graph is needed to make sure that the robber can not re-enter the
reset mechanism, going backwards from his starting position instead of
completing the game.

Our technique is kin to the directed graph case. In particular, we reduce
QBF to our extended variant of~\CR, which in turn reduces to~\CR. We are
not able to really simulate directed graphs, however our construction
produces a graph which behaves as if it were directed, in the sense that
both players can go the wrong way, but they have nothing to gain by doing
it. Our graph is divided into levels arranged in a circular fashion, so
that each level is connected just to the adjacent ones. The pursuit takes
place while players go around the circle: if a given formula is true then
the robber can be captured before he completes a single lap, if it is
false then he can keep circling forever. At diametrical opposite places on
the circle, we have two reset mechanisms, that can be used by either
player to synchronize with his opponent. The obstacle to proving
{\rm EXPTIME}-hardness (hence completeness) by the same method is evident from the description
above: all the interesting action in our graph takes place in a
polynomially bounded number of moves (actually sub-linearly bounded, if we
take every simulation step into account).

\def\V{\mathrm{V}}
The~\textit{number of cops} needed to capture the robber has received
intense study due to the long-standing conjecture of Meyniel that
$O(\sqrt{\lvert\V G\rvert})$~cops suffice on the graph~$G$, where
$\lvert\V G\rvert$ denotes the number of vertices of~$G$---this
conjecture is still open, and the best bound so far
is~$\lvert\V G\rvert2^{-(1+o(1))\sqrt{\log\lvert\V G\rvert}}$,
obtained recently~\cite{LuPeng11,ScoSu11}.
On the
other hand, the~\textit{number of turns} has been considered just
recently~\cite{Boetal09}, obtaining, for some classes of graphs, upper and lower
bounds on the length of the longest games which are linear in the number of
vertices.  Clearly, if the length of any given game of~\CR, provided that
it is a win for the cops and that it is played optimally, could be bounded
by a polynomial in the size of the graph, then, as a corollary of
Theorem~\ref{th-main}, \CR\ would be complete for~{\rm PSPACE}.
However, even though no
super-linear lower bound on the length of the longest games is known to
the author, we believe that our techniques may lead to a proof of
completeness of~\CR\ for~{\rm EXPTIME}. This would imply a
super-polynomial lower bound on the length of the longest games.

\subsection*{Acknowledgements}

The results in this paper first appeared in the author's Master's
thesis~\cite{Mam04}, supervised by Alessandro Berarducci.
The author would like to thank Alessandro Berarducci for suggesting the
problem and reading a preliminary version of this paper. The author is also
obliged to Eleonora Bardelli for comments on a preliminary version, and
for encouraging him to write the paper in the first place.

\vfill\pagebreak
\section{Statement of the result \& outline of the proof}

\CR\ is the problem of determining which
player has a winning strategy in a given instance of the
following game. An instance of the game is described by an undirected
graph~$G$ and a natural number~$n$ smaller than the number~$\lvert\V
G\rvert$ of vertices
of~$G$. The game is played by two contenders, \textit{Cops}
and~\textit{Robber}, with perfect information, by
moving tokens on the vertices of~$G$. At the first move, Cops
places $n$ tokens,~\textit{the cops}, each on a vertex of the graph, multiple
cops are allowed on the same vertex. Then Robber places a single
token,~\textit{the robber}, on some vertex. After that, the
players take turns at moving their tokens from vertex to vertex. At his
turn, each player can move along an edge of the graph each of his tokens,
or leave it unmoved---each token moves at most once per turn, all the
cops can move simultaneously.
Cops wins if at any
time the robber is on the same vertex as one of the cops (the robber has
been~\textit{captured}). Robber wins
if the robber escapes perpetually---or, equivalently, Robber wins after
$\lvert\V G\rvert^{n+1}$~moves have been played,
so the game is not actually infinite.

We will prove the following result.
\begin{theorem}\label{th-main}
\CR\ is hard for {\rm PSPACE}.
\end{theorem}

In the three sections that follow, we will prove Theorem~\ref{th-main}.
Precisely, in Section~\ref{sec-CRp}, we will define a pursuit game and
show that the new game and~\CR\ are mutually {\rm LOGSPACE}-reducible.
Then, in Section~\ref{sec-CRps}, we will describe a construction proving
that satisfiability for boolean formul\ae\ with quantifiers is {\rm
LOGSPACE}-reducible to our new game assuming fixed initial positions of
cops and robber.  Finally, Section~\ref{sec-proof} will use the
construction of Section~\ref{sec-CRps} to conclude the proof of
Theorem~\ref{th-main}.  Although we will try to be as formally correct as
reasonable, we believe that some abuses of nomenclature actually improve
readability.  In particular, the words \textit{cops} and \textit{robber}
are going to be used interchangeably for the tokens and for the vertices
upon which they reside; the same symbol is going to denote an instance of
the problem~\CR, the game constituting that instance, and the graph whereon
that game is played; \textit{\&c}.

\section{Cops and Robber-with-protection}
\label{sec-CRp}

\def\CRp{\textsc{C\&Rp}}

\textsc{Cops and Robber-with-protection} (\CRp) is a variant of~\CR\
played on an undirected graph whose edges are labelled as
either~\textit{protected} or~\textit{unprotected}. The rules for placing
and moving tokens are the same as in~\CR, in
particular tokens can move along any edge irrespective of its label. The
victory condition, on the other hand, changes as follows: Cops wins when
he moves a cop to the vertex occupied by the robber~\textit{through an
unprotected edge}, Robber wins if the robber escapes perpetually.

In the following, we will need to draw diagrams of labelled graphs
for~\CRp. Protected and unprotected edges will be represented by broken
($\xymatrix@1{\ar@<-.5ex>@{--}[]+<2em,1ex>&}$\kern-.4em)
and continuous
($\xymatrix@1{\ar@<-.5ex>@{-}[]+<2em,1ex>&}$\kern-.4em)
lines respectively.
Notice that the presence of multiple edges or loops has no influence on
the classical~\CR\ game, therefore, in this context, graphs are often
assumed to be simple (or reflexive, depending on how the rules of the game
are stated): for us, there is clearly no loss of generality.
On the contrary, in the context of~\CRp, multiple edges can be neglected, but adding
a loop has no consequences only as long as the loop is protected.
Therefore we allow for loops in our graphs, and
we call~\textit{unprotected vertex} a vertex connected to
itself by an unprotected edge, and~\textit{protected vertex} a vertex
which is not unprotected. Protected and unprotected vertices will be
represented by empty ($\circ$) and full ($\bullet$) circles respectively.
Should the reader prefer, he can think at~\CRp\ as defined using just
simple graphs, where both edges and vertices are labeled as
protected or unprotected.
In this case, he should add to the victory condition that if both a cop
and the robber happen to reside on the same vertex at once, with the cop
about to move, then that cop can capture the robber if and only if the
vertex is unprotected.

Observe that we can easily reduce \CR\ to \CRp\ by simply declaring all
edges and all vertices unprotected. The rest of this section will be
devoted to prove the converse.

\begin{lemma}\label{th-CRp}
\CRp\ is {\rm LOGSPACE}-reducible to~\CR.
\end{lemma}
\begin{proof}
Let an instance of~\CRp\ be given by a labelled graph~$G$ and a number~$n$ of
cops. We will show how to construct a graph~$G'$ such that the instance
of~\CR\ described by~$G'$ and~$n$ is a win for Cops if and only if the given
instance of~\CRp\ is a win for Cops.

More precisely, let $\{v_1\dotsc v_g\}$ be the vertices of~$G$. Then
the vertices of $G'$ will be partitioned into $g$ subsets~$G'_1\dotsc
G'_g$ so that the game of~\CR\ played on~$G'$ simulates the game of~\CRp\ 
played on~$G$ in the following sense.
\myitem{a}\label{cond-homo}%
The projection $\pi\colon \V G'\to \V G$ mapping $G'_i$ to~$v_i$ is a
graph homomorphism (neglecting the labels).
\myitem{b}\label{cond-robberiso}%
If the projection of the robber can reach in one move a vertex~$v_i$ of~$G$
which is not threatened (according to the rules of~\CRp) by any of the
projections of the cops, then the robber can reach in one move a vertex
in~$G'_i=\pi^{-1}(v_i)$ which is not threatened (according to the rules of~\CR)
by any cop in~$G'$.
\myitem{c}\label{cond-copsiso}%
If the projection of the robber can not escape capture on the next move by
the projections of the cops in~$G$ (playing~\CRp), then the robber can not
escape capture on the next move in~$G'$ (playing~\CR).

Properties \ref{cond-homo}--\ref{cond-copsiso} imply our claim that $G$ is a win for
$n$~cops in~\CRp\ if and only if $G'$ is a win for $n$~cops in~\CR.
To prove this fact, we consider a game of~\CR\ played on~$G'$, and we look
at the~\textit{projected game} of~\CRp\ on~$G$, \textit{i.e.}\ the game
in which cops and robber are on the projections on~$G$ of cops and robber
in~$G'$.
Because of
property~\ref{cond-homo}, all the moves played in the projected game are
legal.
Assume that there is a perpetually escaping \CRp-strategy for Robber
in~$G$, then he can successfully play~\CR\ in $G'$ by applying his \CRp-strategy
to the projected game.
In fact, using property~\ref{cond-robberiso},
given a safe move in the
projected game, which we are assuming he always has, Robber can produce a
safe move in~$G'$. On the other hand,
the same technique of looking at the projected game
works for Cops, if we assume that \textit{he} has a \CRp-strategy for~$G$,
by property~\ref{cond-copsiso}.

Now we proceed with the construction of~$G'$. First we need an auxiliary
graph~$H$ enjoying the following property. \textit{Robber has
a winning strategy for the instance of \CR\ played by $n$~cops on~$H$.
Moreover,
the vertices of~$H$ are
partitioned into $g$~subsets $H_1\dotsc H_g$ in such a way that,
for any position of cops and robber in which the robber is not captured,
and for any $i=1\dotsc g$, Robber has a winning move that translates his token into~$H_i$}.
In order to construct~$H$, we start with a graph~$P$ such that the robber
can always escape $n$~cops for every initial position (in which the robber
is not already captured). Such a graph can be constructed in time
polynomial in~$n$ -- for instance, $P$ can be the incidence graph of a
finite projective plane of order the smallest power of two greater than or
equal to~$n$ (see~\cite{Pra10}: the proof of Theorem~4.1 and the preceding discussion). The vertices of~$H$ are~$\V H = \{1\dotsc g\}\times \V P$.
We put and edge between vertices $(i,p)$ and~$(j,q)$ if and only if there
is an edge of~$P$ between $p$ and~$q$ or $i=j$.  Finally we
define~$H_i=\{i\}\times\V P$. Clearly, Robber has the required strategy,
in fact he can evade the cops using just the component~$P$, and is free
to choose the other component of each move.

We can now define~$G'$.
The set of vertices of~$G'$ coincides with the set
of the vertices of~$H$, indeed we let~$G'_i=H_i$. Two vertices $v$ and~$w$
of~$G'$ are connected by an edge when either their projections~$\pi(v)$ and~$\pi(w)$ are
connected by an unprotected edge of~$G$, or there is an edge of~$H$
between $v$ and~$w$ and the projections~$\pi(v)$ and~$\pi(w)$ either coincide or are
connected by a protected edge of~$G$. To state it differently, we put in
$G'$ all the edges of~$H$ whose projections are either loops or protected
edges of~$G$, then we add all possible edges~$e$ such that the projection
of~$e$ is either an unprotected edge or a loop at an unprotected vertex.

Remains to prove that $G'$ satisfies properties \ref{cond-homo}--\ref{cond-copsiso}.
\myitem{\ref{cond-homo}} Immediate from the definition.
\myitem{\ref{cond-robberiso}} If the projection of the robber can reach a
safe vertex~$v_i$ in~$G$, then no projection of a cop is connected
to~$v_i$ by an unprotected edge, hence those cops that can
reach~$G'_i$ can do it just moving through edges of~$H$. On the
other hand, since the projection of the robber can reach~$v_i$, then the
robber can reach~$G'_i$, and it can do that moving trough any edge in~$H$.
Follows from the construction of~$H$ that the robber can reach a vertex of~$G'_i$ avoiding all immediate threats.
\myitem{\ref{cond-copsiso}} If the projection of the robber is doomed to be captured at the
next move, then all the vertices reachable from its current position
in~$G$ are connected to the projection of a cop by an unprotected edge.
By~\ref{cond-homo} the projection of the robber must reside on one of this
vertices at the next move, call it $v_i$. When its projection is in~$v_i$,
the robber is in~$G'_i$.
Since $v_i$ is connected to the projection of a cop by an
unprotected edge of~$G$, then all vertices in~$G'_i$ are connected to a cop
by an edge of~$G'$. Hence all vertices reachable by the robber are under
the immediate threat of some cop.

Finally, it is standard to verify that the construction of~$G'$ can be
carried out in~{\rm LOGSPACE}.
\end{proof}

The following corollary was first proven in~\cite{Foetal10}, and also it is,
\textit{a fortiori}, a consequence of our main result. Nevertheless, it
provides an interesting example, since its proof is greatly simplified by
the use of~\CRp.

\begin{corollary}\label{th-NP}
\CR\ is {\rm NP}-hard.
\end{corollary}
\begin{proof}
By reduction of DOMINATING-SET. Let $G$ be any graph, construct the
labelled graph~$G'$ as follows. The set of vertices of~$G'$ coincides with
the set of vertices of~$G$. All vertices of~$G'$ are made unprotected.
The edges of~$G'$ form a complete graph, those edges that are in~$G$ are
labelled as unprotected, the others are protected.

We show that there is a dominating set of $n$ vertices of~$G$ if and only if $n$ cops
can win on~$G'$. In fact, if there is a dominating set, then the cops
can be placed on that set at the first move, and wherever the robber shows
up, it will be captured at the next move. If there is no dominating set,
then at any move there is at least one unthreatened vertex, and the robber can
be moved to it since $G'$ is a complete graph.
\end{proof}

\section{\CRp\ with a given starting position}
\label{sec-CRps}

In this section we will prove the PSPACE-hardness of a simplified version of~\CRp\ where the
starting position of all the tokens is fixed (instead of being decided by
the players at their respective first moves). The result is uninteresting in itself,
since \CR\ with given starting position is already known to be complete for
EXPTIME~\cite{GoRe95}. Nevertheless, in order to prove Theorem~\ref{th-main},
we will exploit some peculiarity of the particular graph constructed below.

\def\CRps{$\textsc{C\&Rp}^\star$}
Let us define~\CRps\ as the problem of deciding what player has a winning
strategy in our simplified game of~\CRp. More precisely an instance
of~\CRps\ is given by a graph~$G$ whose edges are labelled
as~\textit{protected} or~\textit{unprotected}, a natural number~$n$, an
$n$-element multiset of vertices of~$G$ for the Cops' tokens, and a
distinguished vertex of~$G$ for the Robber's token. The problem is to
determine whether Robber has a winning strategy for the game of~\CRp,
played on~$G$, starting with the tokens on the specified positions, Robber
moves first.

\begin{remark}
When dealing with games, in common language, one often talks about
how a game is played
when both players follow~\textit{their best strategies}.
However, the concept is
not well defined, since apparently all moves are equally good for the player which has no
winning strategy. Nevertheless, we will often say
that such and such is going to happen when the game is \textit{played
correctly}, meaning that
\myitem{} what constitutes a best strategy will be~\textit{fixed
conventionally}, by explicit construction,
\myitem{} and it will be proven that deviation from the stipulated best
strategy will not afford any advantage, \textit{i.e.} if any player has a
winning strategy then the official strategy is a winning strategy for him.
\end{remark}

\def\vtx{*+<10pt>[o]{\bullet}}
\def\pvtx{*+<-1pt>[o]{\circ}}
\def\vtxlabel#1#2#3{#1^{\mathrlap{#2}}_{\mathrlap{#3}}}
\begin{lemma}\label{th-CRps}
{\rm QBF} is {\rm LOGSPACE}-reducible to~\CRps.
\end{lemma}
\begin{proof}
\begin{figure}
\footnotesize{
\[\xymatrix@=11pt{
	&
	\save[]+<0pt,3ex>*{\text{\hbox to 0pt{\hss $2i$-th cop's track\hss}}}\restore
	\ar@{.}[d]+0&&*=<20pt>{}\ar@{--}[]+<4pt,0pt>;[dddddddd]+<4pt,0pt>&
	\ar@{.}[rd]+0 &
	\save[]+<0pt,3ex>*{\text{\hbox to 0pt{\hss robber's track\hss}}}\restore
	\vdots & \ar@{.}[ld]+0 & \ar@{--}[dddddddd]
	& && &
	\save[]+<0pt,3ex>*{\text{\hbox to 0pt{\hss $(2i-1)$-th cop's track\hss}}}\restore
	& &\ar@{.}[d]+0& \\
	&\vtx
		\save[]+R+<0pt,0pt>*{\vtxlabel{c}{2i}{2i-3}}\restore
		\ar@/^0pt/@{-}[]+0;[d]+0&&&
	& \vtx\ar@/^0pt/@{-}[]+0;[ld]+0\ar@/^0pt/@{-}[]+0;[rd]+0
		\save[]+R+<1pt,-1pt>*{\vtxlabel{r}{}{2i-3}}\restore&&
	& && && &\vtx
		\save[]+R+<0pt,0pt>*{\vtxlabel{c}{2i-1}{2i-3}}\restore
		\ar@/^0pt/@{-}[]+0;[d]+0& \\
	&\vtx
		\save[]+R+<0pt,0pt>*{\vtxlabel{c}{2i}{2i-2}}\restore
		\ar@/^0pt/@{-}[]+0;[d]+0&&&
	\vtx\ar@/^0pt/@{-}[]+0;[rd]+0
		\save[]+U+<0pt,2pt>*{\mathllap{{r}^{T}_{2i-2}}}\restore
		& & \vtx\ar@/^0pt/@{-}[]+0;[ld]+0
		\save[]+U+<0pt,2pt>*{\vtxlabel{r}{F}{2i-2}}\restore &
	& && && &\vtx
		\save[]+R+<0pt,0pt>*{\vtxlabel{c}{2i-1}{2i-2}}\restore
		\ar@/^0pt/@{-}[]+0;[d]+0& \\
	&\vtx
		\save[]+R+<0pt,0pt>*{\vtxlabel{c}{2i}{2i-1}}\restore
		\ar@/^0pt/@{-}[]+0;[d]+0&&&
	& \vtx\ar@/^0pt/@{-}[]+0;[ld]+0\ar@/^0pt/@{-}[]+0;[rd]+0
		\save[]+R+<1pt,-1pt>*{\vtxlabel{r}{}{2i-1}}\restore&&
	& && && &\vtx
		\save[]+R+<0pt,2pt>*{\vtxlabel{c}{2i-1}{2i-1}}\restore
		\ar@/^0pt/@{-}[]+0;[dl]+0\ar@/^0pt/@{-}[]+0;[dr]+0& \\
	&\vtx
		\save[]+R+<0pt,2pt>*{\vtxlabel{c}{2i}{2i}}\restore
		\ar@/^0pt/@{-}[]+0;[ld]+0\ar@/^0pt/@{-}[]+0;[rd]+0&&&
	\vtx\ar@/^0pt/@{-}[]+0;[rd]+0
		\save[]+UL+<5pt,2pt>*{\mathllap{{r}^{T}_{2i}}}\restore
		& & \vtx\ar@/^0pt/@{-}[]+0;[ld]+0
		\save[]+U+<0pt,2pt>*{\vtxlabel{r}{F}{2i}}\restore &
	& && && \vtx
		\save[]+R+<0pt,-2pt>*{\vtxlabel{T}{2i-1}{2i}}\restore
		\ar@/^0pt/@{-}[]+0;[d]+0&&
		\vtx
		\save[]+R+<0pt,0pt>*{\vtxlabel{F}{2i-1}{2i}}\restore
		\ar@/^0pt/@{-}[]+0;[d]+0 \\
	\vtx
		\save[]+R+<0pt,-2pt>*{\vtxlabel{T}{2i}{2i+1}}\restore
		\ar@/^0pt/@{-}[]+0;[d]+0&&
		\vtx
		\save[]+R+<0pt,0pt>*{\vtxlabel{F}{2i}{2i+1}}\restore
		\ar@/^0pt/@{-}[]+0;[d]+0&&
	& \vtx\ar@/^0pt/@{-}[]+0;[ld]+0\ar@/^0pt/@{-}[]+0;[rd]+0
		\save[]+R+<1pt,-1pt>*{\vtxlabel{r}{}{2i+1}}\restore&&
	& \vtx
		\save[]+R+<0pt,-6pt>*{\vtxlabel{a}{2i-1}{2i+1}}\restore
		\ar@/^0pt/@{-}[]+0;[d]+0\ar@/^0pt/@{-}[]+0;[urrrr]+0\ar@/^0pt/@{-}[]+0;[ullll]+0
	&&\vtx
		\save[]+R+<0pt,-6pt>*{\vtxlabel{b}{2i-1}{2i+1}}\restore
		\ar@/^0pt/@{-}[]+0;[d]+0\ar@/^0pt/@{-}[]+0;[urrrr]+0\ar@/^0pt/@{-}[]+0;[ullll]+0
	&&\vtx
		\save[]+R+<0pt,0pt>*{\vtxlabel{T}{2i-1}{2i+1}}\restore
		\ar@/^0pt/@{-}[]+0;[d]+0&&
		\vtx
		\save[]+R+<0pt,0pt>*{\vtxlabel{F}{2i-1}{2i+1}}\restore
		\ar@/^0pt/@{-}[]+0;[d]+0 \\
	\vtx
		\save[]+R+<0pt,0pt>*{\vtxlabel{T}{2i}{2i+2}}\restore
		\ar@/^0pt/@{-}[]+0;[d]+0&&
		\vtx
		\save[]+R+<0pt,0pt>*{\vtxlabel{F}{2i}{2i+2}}\restore
		\ar@/^0pt/@{-}[]+0;[d]+0&&
	\vtx\ar@/^0pt/@{-}[]+0;[rd]+0
		\save[]+U+<0pt,2pt>*{\mathllap{{r}^{T}_{2i+2}}}\restore
		& & \vtx\ar@/^0pt/@{-}[]+0;[ld]+0
		\save[]+U+<0pt,2pt>*{\vtxlabel{r}{F}{2i+2}}\restore &
	& \vtx
		\save[]+R+<0pt,0pt>*{\vtxlabel{a}{2i-1}{2i+2}}\restore
		\ar@/^0pt/@{-}[]+0;[d]+0&&
		\vtx
		\save[]+R+<0pt,0pt>*{\vtxlabel{b}{2i-1}{2i+2}}\restore
		\ar@/^0pt/@{-}[]+0;[d]+0 &&
	\vtx
		\save[]+R+<0pt,0pt>*{\vtxlabel{T}{2i-1}{2i+2}}\restore
		\ar@/^0pt/@{-}[]+0;[d]+0&&
		\vtx
		\save[]+R+<0pt,0pt>*{\vtxlabel{F}{2i-1}{2i+2}}\restore
		\ar@/^0pt/@{-}[]+0;[d]+0 \\
	\vtx
		\save[]+R+<0pt,0pt>*{\vtxlabel{T}{2i}{2i+3}}\restore
		\ar@{.}[]+0;[d]&&
		\vtx
		\save[]+R+<0pt,0pt>*{\vtxlabel{F}{2i}{2i+3}}\restore
		\ar@{.}[]+0;[d]&&
	& \vtx\ar@{.}[]+0;[ld]\ar@{.}[]+0;[rd]
		\save[]+R+<1pt,-1pt>*{\vtxlabel{r}{}{2i+3}}\restore&&
	& \vtx
		\save[]+R+<0pt,0pt>*{\vtxlabel{a}{2i-1}{2i+3}}\restore
		\ar@{.}[]+0;[d]&&
		\vtx
		\save[]+R+<0pt,0pt>*{\vtxlabel{b}{2i-1}{2i+3}}\restore
		\ar@{.}[]+0;[d] &&
	\vtx
		\save[]+R+<0pt,0pt>*{\vtxlabel{T}{2i-1}{2i+3}}\restore
		\ar@{.}[]+0;[d]&&
		\vtx
		\save[]+R+<0pt,0pt>*{\vtxlabel{F}{2i-1}{2i+3}}\restore
		\ar@{.}[]+0;[d] \\
	&\vdots&&&
	& \vdots & &
	& \bigstar&\vdots&\bigstar && &\vdots& 
}\]}
\caption{\hbox to 0pt{\label{fig-CRps}}% without hbox I got an error: showkeys' mistake?
Seven levels of three tracks of~$G_\Phi$}
\end{figure}
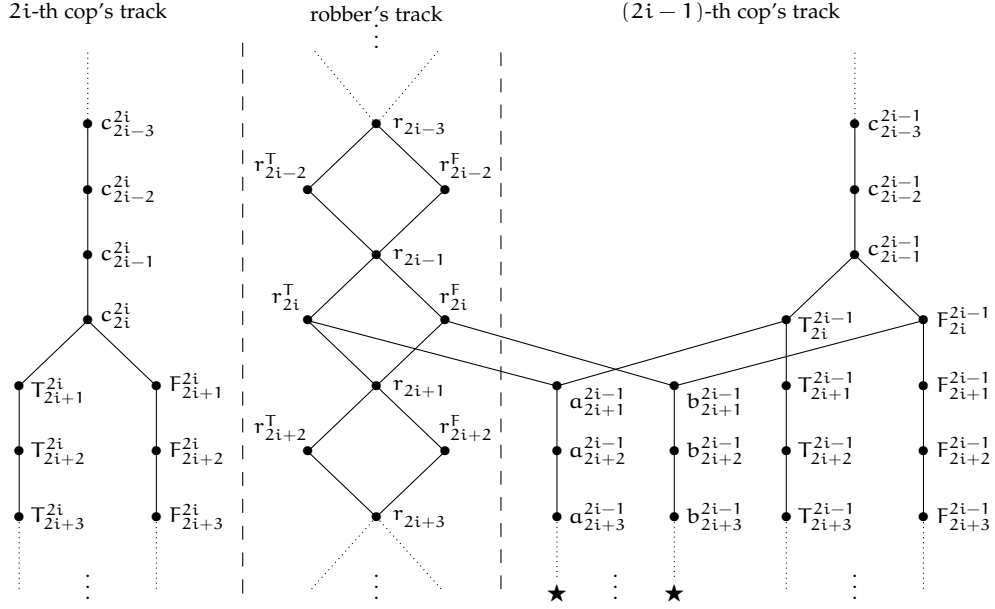
Let a quantified boolean formula be given
\[
\Phi = \forall v_1\exists v_2\dotsc\forall v_{2n-1}\exists
v_{2n}\phi(v_1\dotsc v_{2n})
\]
with $\phi(v_1\dotsc v_{2n})$ quantifier free and in conjunctive normal
form.
We will construct a graph~$G_\Phi$, and show an appropriate initial
position of cops and robber, such that Cops has a winning strategy for the
game~\CRps\ if and only if $\Phi$ is true.

To make the construction clear,
it is best to start by describing the~\textit{geography} of~$G_\Phi$, and
only after that to go into the details. First of all, the reader should
take a look at Figure~\ref{fig-CRps}, which depicts part of~$G_\Phi$.
The graph~$G_\Phi$ will be divided into two~\textit{stages}, each
constituted by several~\textit{levels}---vertices aligned
horizontally in Figure~\ref{fig-CRps} are on the same level.
Each level is
connected just to the level above and to the level below, with no edge
connecting a vertex of a level to another vertex of the same level.
We will place a total of~$2n+2$ cops on this graph.

The first
stage lets Robber and Cops alternatively choose the value of the
universally and existentially quantified variables respectively. These
boolean values will be stored in the positions of $2n$ cops.
This stage takes $2n$ levels, one for each of the variables.
The first stage will be divided into~\textit{tracks} as well---tracks are represented
vertically in Figure~\ref{fig-CRps}. Assuming correct play, we will have
the robber and $2n$ cops starting at level~$1$, each on a separate track.
Each of this tokens will move down one level per move, always remaining
on~\textit{its} track. After move~$2n$, all the tokens will be at
level~$2n+1$, which is the first level of the second stage. At this level
we will have two vertices for each of the variables, precisely one of
which occupied by a cop, and a single vertex for the robber.
The second stage will simply compute the value of~$\phi$. At the last
level of the second stage we have a safe heaven: a protected vertex connected to the
level above by a protected edge. If the robber reaches this vertex, it will
never be captured. If we except the heaven and the edges leading to it,
all other vertices and all other edges will be unprotected. Indeed, our
construction works as well in~\CR, replacing the safe heaven, for example, by a finite projective plane.

Now to the description of the first stage.
Figure~\ref{fig-CRps} represents, in particular, part of this stage, and the
reader is referred to it.
\myitem{} The robber's track is constituted by a single vertex~$r_{2i+1}$ for each
odd-numbered level, and two vertices~$r^T_{2i}$ and~$r^F_{2i}$ for each
even-numbered level. All these vertices are unprotected. For all $i$, the
vertices~$r^T_{2i}$ and~$r^F_{2i}$ are connected by unprotected edges to
the vertex~$r_{2i-1}$ above and to the vertex~$r_{2i+1}$ below.
The idea is that, while running down the track, the robber will determine the
value of universally quantified variables by deciding through which side of
each diamond to travel. Variable $v_{2i-1}$ is assigned at Robber's
$(2i-1)$-th move.
\myitem{} Cops meant to store the values of existentially quantified
variables run on a simple linear sequence of vertices, one per
level, each connected to the one on the level below.
Let $c^{2i}_j$ denote the vertex on level~$j$ of the track assigned
to variable~$v_{2i}$. After level~$2i$, the track of variable~$v_{2i}$ bifurcates
into two linear sequences of vertices denoted~$T^{2i}_j$
and~$F^{2i}_j$. The starting vertices~$T^{2i}_{2i+1}$ and~$F^{2i}_{2i+1}$
of these new sequences are both connected to $c^{2i}_{2i}$.
Again, all vertices and all edges are unprotected. The bifurcations allow
Cops to select the value of existentially quantified variables, precisely
$v_{2i}$ will be fixed at Cop's $2i$-th move.
\myitem{} Cops meant for universally quantified variables run on similar
tracks. In particular, the track assigned to~$v_{2i-1}$ bifurcates after
level~$2i-1$. However, the branch taken by the cop assigned to each of these
tracks at the bifurcation
must be determined by the position of the robber. To this aim, we connect
$T^{2j-1}_{2j}$ and~$r^T_{2j}$, through unprotected edges, to an unprotected
vertex~$a^{2j-1}_{2j+1}$ placed at level~$2j+1$, and in turn this vertex to
a safe heaven. This way, if, after
Robber's move~$2j-1$, the robber is in~$r^T_{2j}$, then Cops, who plays his
own $(2j-1)$-th move after Robber's one, is forced to move the cop which is
in~$c^{2i-1}_{2i-1}$ to~$T^{2j-1}_{2j}$, otherwise nothing would stop the
robber from reaching the safe heaven through~$a^{2j-1}_{2j+1}$. Because
of how we are going to use the graph~$G_\Phi$ in the proof of
Theorem~\ref{th-main}, we need to place the safe heaven at the end of a
linear sequence of vertices spanning all the levels of~$G_\Phi$: this
detail is immaterial for the proof at hand. Finally, we construct a
similar device for vertices~$F^{2j-1}_{2j}$ and~$r^F_{2j}$, connecting
them to~$b^{2j-1}_{2j+1}$.

The second stage begins at level~$2n+1$. At this level we have two
vertices~$T^i_{2n+1}$ and~$F^i_{2n+1}$ for each variable~$v_i$, plus one
vertex~$r_{2n+1}$, and a few more of the vertices denoted by~$a$ and~$b$. As
explained, the $a$ and~$b$ vertices, now, can be neglected. Assuming
correct play, just before move~$2n+1$, we have the robber in~$r_{2n+1}$
and precisely one cop in each pair of vertices~$T^i_{2n+1}$
and~$F^i_{2n+1}$: the presence of this cop in~$T^i_{2n+1}$ denotes truth
of~$v_i$, the presence of the cop in~$F^i_{2n+1}$ denotes falsity.
At level~$2n+2$ we place one unprotected vertex for each clause
of~$\phi$. The vertex associated to clause~$c$ will be connected by
an unprotected edge to~$T^i_{2n+1}$ whenever $v_i$ is in~$c$, and
to~$F^i_{2n+1}$ whenever $\neg v_i$ is in~$c$. All clauses are connected
by unprotected edges to~$r_{2n+1}$. Clearly, at move~$2n+1$, the robber can
be moved safely to one of the clauses' vertices only if that clause is
false. Finally we connect all clauses to a safe heaven, so that the
robber can reach it if and only if the formula is false.

To complete the construction of~$G_\Phi$, we attach a single unprotected
vertex~$c_0$, by an unprotected edge, to $r_1$. That vertex is on
level~$0$, and we place two cops on it at the start of the game.

Now we prove that if $\Phi$ is true, then Cops has a winning strategy.
First of all, observe that Cops can force the robber to go down its track,
one level per move, until it reaches the second stage. To do so, he will
use the two cops initially placed in~$c_0$ to completely occupy the
robber's track precisely one level behind the robber, so that Robber has
to move his token one level down each move in order to avoid capture. The
only way out of the track is through vertices named with~$a$ and~$b$, so
Cops will act as described above in order to block this escape. By
following this strategy, Cops lets Robber choose the value of universally
quantified variables at odd-numbered moves, and he can choose the value of
existentially quantified variables at even-numbered moves. Observe that
the order of choices coincides with the order of the quantifiers. Since
$\Phi$ is true, Cops can make his choices so that $\phi(v_1\dotsc
v_{2n})$ is true. Hence the robber will be captured as soon as
it enters one of the clauses.

Conversely, we prove that if $\Phi$ is false, then Robber has a winning
strategy. Our strategy will move the robber down one level per move,
hence no matter how they move, the two cops initially placed in~$c_0$ will
never be able to capture it. By the same reason, no cop moving backwards,
or standing for one move,
can capture the robber, so we can assume that all cops move down one level
per move, and, in particular, none of them can leave its track. By the
usual reason, Robber can force the value of odd-numbered variables. Now,
by our assumption, Cops has chosen values to even-numbered variables at
proper times. However, since $\Phi$ is false, Robber can make his choices so
that $\phi(v_1\dotsc v_{2n})$ will be false, hence he will be able to move
the robber to an unthreatened clause, whence it will reach safety.
\end{proof}

\section{Proof of the main theorem}
\label{sec-proof}

In this section we will prove Theorem~\ref{th-main}. Our technique is,
again, by reduction of~QBF to~\CRp.
In particular, we will connect a few copies of the graph~$G_\Phi$
constructed in the previous section to two~\textit{reset mechanisms}---a
portion of one of which is shown in Figure~\ref{fig-reset}. The function
of the reset mechanism is to substitute the safe heaven at the end
of~$G_\Phi$ \textit{and} to allow either player to force the initial
position. As we will see, the mechanisms have been devised so that
the robber can safely inhabit either of them unless all the cops
are employed to chase it through a very specific maneuver,
and doing that the initial position of the proof of Lemma~\ref{th-CRps} is
attained. If $\Phi$ happens to be false, then Robber can move his token
safely trough~$G_\Phi$ to the other reset mechanism.
On the other hand, we will see that Cops as well has means to force Robber to
get into the starting position. So Cops will have a winning strategy in
our instance of~\CRp\ if and only if he has a strategy for~\CRps\ on the
graph~$G_\Phi$ if and only if $\Phi$ is true.

The argument that follows and the proof of Lemma~\ref{th-CRps} are
actually one single proof. We decided to separate a substantial portion of
it into Lemma~\ref{th-CRps} in order to give a more orderly exposition.
Nevertheless the reader will not understand the rest of this section
unless he diligently went trough Section~\ref{sec-CRps} before.

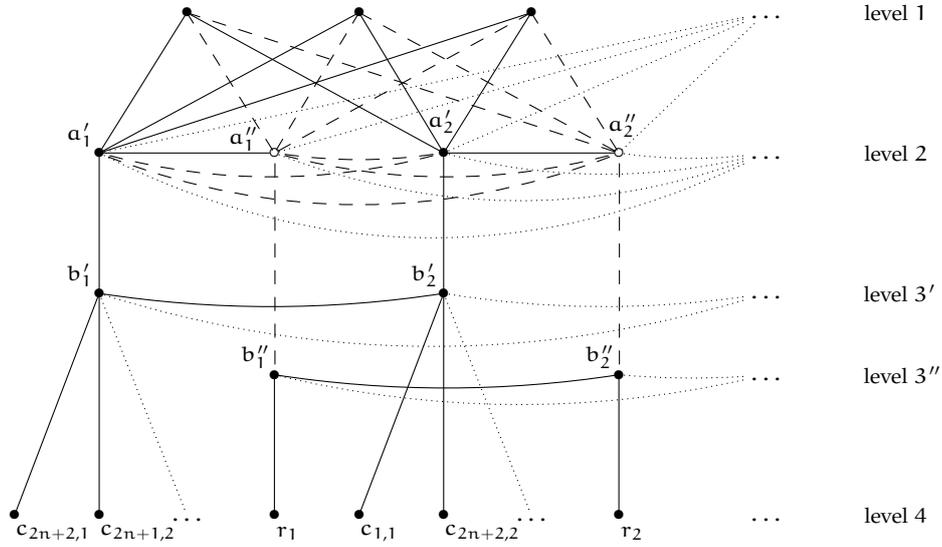
\begin{figure}
\footnotesize{
\[\xymatrix@=17pt{
	&& \vtx\ar@{-}@/^0pt/[]+0;[ddl]+0
		\ar@{--}[]+0;[ddr]
		\ar@{-}@/^0pt/[]+0;[ddrrr]+0
		\ar@{--}[]+0;[ddrrrrr]
	&& \vtx\ar@{-}@/^0pt/[]+0;[ddlll]+0
		\ar@{--}[]+0;[ddl]
		\ar@{-}@/^0pt/[]+0;[ddr]+0
		\ar@{--}[]+0;[ddrrr]
	&& \vtx\ar@{-}@/^0pt/[]+0;[ddlllll]+0
		\ar@{--}[]+0;[ddlll]
		\ar@{-}@/^0pt/[]+0;[ddl]+0
		\ar@{--}[]+0;[ddr]
	&&& \dotso
		\ar@{.}[];[ddllllllll]+0
		\ar@{.}[];[ddllllll]
		\ar@{.}[];[ddllll]+0
		\ar@{.}[];[ddll]
	&*[r]{\text{level $1$}}
\\
\\
	&\vtx\ar@/^0pt/@{-}[]+0;[rr]
		\save[]+UL+<0pt,0pt>*{a'_1}\restore
		\ar@{--}@/_8pt/[]+<4pt,-1pt>;[rrrr]+<-4pt,-1pt>
		\ar@{--}@/_18pt/[]+<3.5pt,-1.5pt>;[rrrrrr]+<-3.5pt,-1.5pt>
		\ar@{.}@/_32pt/[]+0;[rrrrrrrr]
	&&\pvtx
		\save[]+UL+<-10pt,4pt>*{a''_1}\restore
		\ar@{--}@/_2pt/[]+<4pt,-.5pt>;[rr]+<-4pt,-.5pt>
		\ar@{--}@/_8pt/[]+<4pt,-1pt>;[rrrr]+<-4pt,-1pt>
		\ar@{.}@/_18pt/[];[rrrrrr]
	&&\vtx\ar@/^0pt/@{-}[]+0;[rr]
		\save[]+U+<-1pt,5pt>*{a'_2}\restore
		\ar@{.}@/_8pt/[]+0;[rrrr]
	&&\pvtx
		\save[]+UL+<3pt,9pt>*{a''_2}\restore
		\ar@{.}@/_2pt/[];[rr]
	&&\dotso
	&*[r]{\text{level $2$}}
\\
\\
	&\vtx\ar@/^0pt/@{-}[]+0;[uu]+0
		\save[]+UL+<0pt,0pt>*{b'_1}\restore
		\ar@{-}@/_5pt/[]+0;[rrrr]+0
		\ar@{.}@/_20pt/[]+0;[rrrrrrrr]
	&&&&\vtx\ar@/^0pt/@{-}[]+0;[uu]+0
		\save[]+UL+<0pt,0pt>*{b'_2}\restore
		\ar@{.}@/_5pt/[]+0;[rrrr]
	&&&& \dotso
	&*[r]{\text{level $3'$}}
\\
	&&&\vtx\ar@/^0pt/@{--}[]+<.2pt,4pt>;[uuu]+<.2pt,-4pt>
		\save[]+UL+<0pt,0pt>*{b''_1}\restore
		\ar@{-}@/_5pt/[]+0;[rrrr]+0
		\ar@{.}@/_11.25pt/[]+0;[rrrrrr]
	&&&&\vtx\ar@/^0pt/@{--}[]+<.2pt,4pt>;[uuu]+<.2pt,-4pt>
		\save[]+UL+<0pt,0pt>*{b''_2}\restore
		\ar@{.}@/_1.25pt/[]+0;[rr]
	&& \dotso
	&*[r]{\text{level $3''$}}
\\
\\
	\vtx\ar@/^0pt/@{-}[]+0;[uuur]+0
		\save[]+D+<3pt,0pt>*{c_{\mathrlap{2n+2,1}}}\restore
	&\vtx\ar@/^0pt/@{-}[]+0;[uuu]+0
		\save[]+D+<3pt,0pt>*{c_{\mathrlap{2n+1,2}}}\restore
	&\dotso\ar@{.}[];[uuul]+0
	&\vtx\ar@/^0pt/@{-}[]+0;[uu]+0
		\save[]+D+<3pt,0pt>*{r_{\mathrlap{1}}}\restore
	&\vtx\ar@/^0pt/@{-}[]+0;[uuur]+0
		\save[]+D+<3pt,0pt>*{c_{\mathrlap{1,1}}}\restore
	&\vtx\ar@/^0pt/@{-}[]+0;[uuu]+0
		\save[]+D+<3pt,0pt>*{c_{\mathrlap{2n+2,2}}}\restore
	&\dotso\ar@{.}[];[uuul]+0
	&\vtx\ar@/^0pt/@{-}[]+0;[uu]+0
		\save[]+D+<3pt,0pt>*{r_{\mathrlap{2}}}\restore
	&&\dotso
	&*[r]{\text{level $4$}}
}\]}
\caption{\hbox to 0pt{\label{fig-reset}}% without hbox I got an error: showkeys' mistake?
A portion of a reset mechanism}
\end{figure}

Let us assume that we have formula
\[
\Phi = \forall v_1\exists v_2\dotsc\forall v_{2n-1}\exists
v_{2n}\phi(v_1\dotsc v_{2n})
\]
as in the proof of
Lemma~\ref{th-CRps}, we are going to construct a labelled graph~$G$ such
that Cops has a winning strategy for $2n+2$ cops if and only if $\Phi$ is
true. We will use $4n+4$ slightly modified copies of the graph~$G_\Phi$ connected to two
reset mechanisms. It is convenient to give an overview of how these
components fit together to make the graph~$G$.
The copies of~$G_\Phi$ are modified by removing the vertex~$c_0$ and all
the safe heavens. So the top level of each copy of~$G_\Phi$ is constituted
by the starting vertices of the robber and all the cops, and the bottom
level is constituted by the vertices representing the clauses of~$\phi$ and
vertices~$a^{\textrm{odd}}_{2n+2}$ and~$b^{\textrm{odd}}_{2n+2}$ with odd
upper index (\textit{i.e.}\ those ends of the cops' tracks that went
directly into the removed safe heavens).
For clarity's sake, the reset mechanisms are divided
into levels as well---as usual, vertices aligned horizontally in
Figure~\ref{fig-reset} are on the same level.
Broadly speaking, we will arrange $2n+2$ copies of~$G_\Phi$ so that
their top levels coincide with the bottom level of one of the reset
mechanisms; the top level of the other reset mechanism will coincide with
the bottom levels of these copies of~$G_\Phi$; then the top level of the
other $2n+2$ copies of~$G_\Phi$ will coincide with the bottom level of
this second reset mechanism; and, to close the circle, the bottom level
of the last $2n+2$ copies of~$G_\Phi$ is going to coincide with the top
level of the first reset mechanism.

More precisely, fix $2n+2$ of our modified copies of~$G_\Phi$. One reset
mechanism is constructed as follows.
\myitem{} On level~$1$, we place all the vertices belonging to the bottom
levels of the $2n+2$ copies of~$G_\Phi$, these vertices include those
representing the clauses of~$\phi$.
\myitem{} On level~$2$, we place $2n+2$ unprotected vertices~$a'_1\dotsc
a'_{2n+2}$ and $2n+2$ protected vertices~$a''_1\dotsc a''_{2n+2}$. For
each~$i=1\dotsc 2n+2$ we put an unprotected edge between~$a'_i$
and~$a''_i$, unprotected edges between~$a'_i$ and all vertices of
level~$1$, and protected edges between~$a''_i$ and all vertices of
level~$1$. Finally, all the pair of vertices of level~$2$ that are not
already joined by an edge, are connected using an unprotected edge, so
that level~$2$ forms a complete graph.
\myitem{} On level~$3'$ we put vertices~$b'_1\dotsc b'_{2n+2}$. Each
$b'_i$ is connected by an unprotected edge to the corresponding~$a'_i$.
All vertices of level~$3'$ are connected together in a complete graph of
unprotected edges.
\myitem{} On level~$3''$ we put vertices~$b''_1\dotsc b''_{2n+2}$. Each
$b''_i$ is connected by a protected edge to the corresponding~$a''_i$.
All vertices of level~$3''$ are connected together in a complete graph of
unprotected edges.
\myitem{} Finally, level~$4$ is constituted by the top levels of the other
$2n+2$ copies of~$G_\Phi$ arranged as follows.  Let~$G_{\Phi,1}\dotsc
G_{\Phi,2n+2}$ be our copies of~$G_\Phi$. Call $r_i$ the starting vertex
for the robber in~$G_{\Phi,i}$ and call~$c_{i,1}\dotsc c_{i,2n}$ its cops'
staring vertices.  For all~$i$, we connect $r_i$ to~$b''_i$ with an
unprotected edge, and $c_{i,1}\dotsc c_{i,2n}$ to~$b'_{i+1}\dotsc
b'_{2n+2}b'_1\dotsc b'_{i-2}$, again with unprotected edges.\\
The second reset mechanism is constructed symmetrically.

We will need to remember that,
by construction, $G$ is divided into $4n+8$ levels. In fact, starting, say, from
level~$1$ of one reset mechanism, we can count level~$2$ of the same mechanism,
then level~$3$, which is the union of levels~$3'$ and~$3''$, then
level~$4$ of the reset mechanism, which is level~$1$ of the copies
of~$G_\Phi$. After this we have the~$2n+2$ levels of the copies
of~$G_\Phi$, \textit{i.e.}\ one level for all the levels~$1$ of them, one
for all the levels~$2$, and so on. At level~$2n+1$ we have enumerated
precisely $2n+4$ different levels, the next level is level~$1$ of the
opposite reset mechanism, and we can define the remaining $2n+4$ levels
symmetrically. The important observation is that a vertex in one level can be connected only to
vertices of the same level or vertices of one of the two neighbouring
levels in the (circular) order of our enumeration.
It is here that we need the linear sequences of vertices indicated by letters~$a$ and~$b$ in
Figure~\ref{fig-CRps} to be no longer nor shorter than the~\textit{normal}
path to safety.

Now we show a winning strategy for Robber assuming that $\Phi$ is false.
Observe that each cop can pose a threat on at most one of the
vertices~$a''_i$, so, after Cops has placed his tokens, there is at least
one unthreatened~$a''_i$ in one of the reset mechanisms. The robber should
be placed in that vertex. Now, let's focus on the reset mechanism in
which the robber has been placed. Until all the cops occupy precisely all
the vertices~$a'_i$ of that reset mechanism, the robber will simply stay
on an unthreatened~$a''_i$---it can be moved between them through the complete
graph of protected edges. When all the cops occupy all the
vertices~$a'_i$, Robber will have his token in, say, $a''_k$. Now he will
send it to~$b''_k$, whence it will enter the $k$-th copy of~$G_\Phi$
connected to that reset mechanism, and arguably emerge unscathed from it
in the opposite reset mechanism. To prove that Robber's plan actually
works observe the following facts.
\myitem{} No cop can capture the robber while it moves from~$a''_k$
to~$b''_k$ to~$r_k$.
\myitem{}
The levels of $G$ are arranged in a circle, where each level is connected
precisely to two neighbouring ones.
The robber starts its journey from
level~$2$ of one of the reset mechanisms, and it travels at the speed of
one level per turn heading for the same level
in the other reset mechanism, which is precisely half way around the
circle. Since Robber moves first, no cop will arrive there before the
robber does.
\myitem{} By the time the robber is in~$r_k$, precisely $2n$ cops can reach
their $2n$ starting vertices in the copy of~$G_\Phi$ that the robber
just entered. The remaining cops, from now on, can be neglected, because
they will not enter our copy~$G_\Phi$ in time.
Hence, following the strategy detailed in the proof of
Lemma~\ref{th-CRps}, the robber will reach one of the clauses of~$\phi$,
escape capture there having chosen the values of the variables properly,
and finally move to level~$2$ of the new reset mechanism before any cop can
be there.

Remains to be proven that if $\Phi$ is true, then Cops has a winning
strategy. In the following, we will assume that $\phi$ has at least $8$
variables (we need to have at least $9$ cops around), and that $\phi$ has
at least one non-empty clause (we need $G$ to be connected). Clearly, this
goes without loss of generality. The intermediate goal of Cops is to reach
the following position, with the Robber about to move:
\myitem{} robber in some~$a''_i$ of one reset mechanism,
\myitem{} $2n+1$ of the cops in vertices~$a'_j$ of the same
reset mechanism, with~$j\neq i$,
\myitem{} and the remaining cop either in~$a'_i$ or in~$a''_i$.\\
To this aim, he places initially three cops in vertices~$b'_1$, $b''_2$,
and~$a'_3$ of each reset mechanism ($6$~cops total), we don't care where
the remaining cops are placed.  The cops at~$a'_3$ make a barrier at
level~$1$ of their respective reset mechanisms, and the cops at~$b'_1$
and~$b''_2$ make a barrier at level~$3$.  Hence, as long as these $6$ cops
stay in place, the graph is effectively divided, from Robber's point of
view, into $4n+6$ disjoint components: the $4n+4$ copies of~$G_\Phi$ and
the two reset mechanisms. We claim that if Robber places his token in one
of the copies of~$G_\Phi$, then Cops can win using three additional cops
(over the $6$ above). The strategy is as follows. Two cops reach
level~$1$ of the robber's track
(meaning the vertices labelled with~$r$,
\textit{not} the track where the robber is currently located)
in the copy of~$G_\Phi$ inhabited by the robber, they
can do that since $G$~is connected, then they move downwards until they
reach the same level as the robber (with the robber about to move). From
this moment on, these two cops will be kept at the same level of the
robber, \textit{i.e.}\ they will be moved up or down whenever the robber
is moved up or down, also they will be placed so that they occupy their
level of the robber's track completely. Hence, form this moment on, the
robber can not access the robber's track, which implies that it must be on
a cop's track and it can not move from this cop's track to another. Since
all cop's track are trees, a single additional cop is sufficient to
capture the robber. Therefore we know that Robber must place his token in one
of the reset mechanisms.

Now we assume that Robber placed the robber in one of the reset
mechanisms.
We will now explain how Cops can attain his intermediate goal.
All the action described in this paragraph takes place inside the reset
mechanism chosen by the robber. Cops has three tokens in vertices~$b'_1$,
$b''_2$, and~$a'_3$. He now moves all the other tokens to vertices
$a'_4\dotsc a'_{2n+2}$. While doing that, the cops in~$b'_1$, $b''_2$,
and~$a'_3$
are not moved (unless the robber tries to escape), so at the end of the
maneuver the robber must be either in~$a''_1$, in~$a'_2$, or in~$a''_2$: the only
unthreatened vertices of the reset mechanism. Now Cops can reach his goal
by one of the following sequences of moves, which are forced for Robber:
\myitem{} If the robber is in~$a''_1$:
cop from~$b'_1$ to~$a'_1$ --
robber in~$a'_2$ or~$a''_2$ --
if the robber is in~$a''_2$: cop from~$b''_2$ to~$a''_2$ and we are
done---otherwise the robber is in~$a'_2$: cop from~$a'_1$ to~$b'_1$ and cop from~$b''_2$ to~$a''_2$ --
robber either in~$a''_1$ or~$a''_2$ --
cop from~$a''_2$ to~$a'_2$ and cop from~$b'_1$ to~$a'_1$.
\myitem{} If the robber is in~$a'_2$:
cop form~$b''_2$ to~$a''_2$ --
robber either in~$a''_1$ or~$a''_2$ --
cop from~$a''_2$ to~$a'_2$ and cop form~$b'_1$ to~$a'_1$.
\myitem{} If the robber is in~$a''_2$:
cop from~$b'_1$ to~$a'_1$ and cop from~$b''_2$ to~$a''_2$.

Finally, the robber is in some~$a''_j$, and all the vertices~$a'_i$ are
occupied by cops except at most~$a'_j$, in which case its cop must be
in~$a''_j$.  Without loss of generality we can assume~$j=2$.  The only
option for Robber is to move his token down to~$b''_2$. Now the
following happens:
\myitem{} all the cops in~$a'_3\dotsc a'_{2n+2}$ are moved down
to~$b'_3\dotsc b'_{2n+2}$,
\myitem{} if $a'_2$ is occupied by a cop, then that cop is not moved,
otherwise there is a cop in~$a''_2$, and that cop is moved to~$a'_2$,
\myitem{} the cop in~$a'_1$ is moved to~$a''_1$.\\
At this point there are two cases, either the robber moves to~$r_2$ or it
moves to one of the vertices~$b''_i$ (here including the case if it stays
in~$b''_2$).
\myitem{} In the first case, the cops in~$b'_3\dotsc b'_{2n+2}$ must be
moved to~$c_{2,1}\dotsc c_{2,2n}$, the cop in~$a'_2$ goes to~$a''_2$, and
the cop in~$a''_1$ goes to~$b''_1$. This way the robber can not be moved
back to~$b''_2$, and it can either be moved into the robber's track of the
second (because it is in~$r_2$) copy of~$G_\Phi$, or be left in~$r_2$. In
either case the cops in~$a''_2$ and~$b''_1$ will be moved to~$b''_2$ at
the next move, forcing it downwards. These two cops have the same function
of the cops in vertex~$c_0$ of the proof of Lemma~\ref{th-CRps}, although
they may be lagging two levels behind the robber instead of one. The
reader can check that the argument of that proof applies from now on to
the second copy of~$G_\Phi$. Hence the robber can not emerge form that
copy of~$G_\Phi$ uncaptured.
\myitem{} In the second case, let's assume that Robber moved his token
to~$b''_j$. Then Cops moves the cops from~$b'_3\dotsc b'_{2n+2}$
to~$b'_{j+1}\dotsc b'_{2n+2}b'_1\dotsc b'_{j-2}$, the cop in~$a'_2$ is
sent to a~$a'_j$, and the cop in~$a''_1$ is sent to~$b''_1$. Now the
robber is forced to~$r_j$ and the strategy of the first case applies.

To conclude the proof, suffices to verify that our graph~$G$ can be constructed
in~{\rm LOGSPACE}, which is standard. \qed

\vfill\pagebreak
\catcode`$=3

%\bibliography{bibliography}
%\bibliographystyle{amsalpha}

\newcommand{\etalchar}[1]{$^{#1}$}
\providecommand{\bysame}{\leavevmode\hbox to3em{\hrulefill}\thinspace}
%\providecommand{\MR}{\relax\ifhmode\unskip\space\fi MR }
%% \MRhref is called by the amsart/book/proc definition of \MR.
%\providecommand{\MRhref}[2]{%
%  \href{http://www.ams.org/mathscinet-getitem?mr=#1}{#2}
%}
%\providecommand{\href}[2]{#2}

\def\MR#1{\MRhref#1 \stop{MR #1}}
\def\MRhref#1 #2\stop#3{\href{http://www.ams.org/mathscinet-getitem?mr=#1}{#3}}
\def\url#1{\href{#1}{\tt #1}}

\end{document}